\documentclass[11pt]{article}
\usepackage{amsmath}
\usepackage{amssymb}
\usepackage{mathrsfs}
\usepackage{algorithm}
\usepackage{algorithmicx}
\usepackage[noend]{algpseudocode}

\usepackage{amsthm}
\usepackage[framemethod=TikZ]{mdframed}
\usepackage{framed}
\usepackage[margin=1in]{geometry}
\usepackage{graphicx}
\usepackage[outdir=./]{epstopdf}
\usepackage{caption}
\usepackage{subcaption}
\usepackage{multirow}
\usepackage{comment}
\usepackage{color}
\usepackage{xcolor}
\definecolor{mydarkblue}{rgb}{0,0.08,0.45}
\usepackage[hidelinks,colorlinks]{hyperref}
\hypersetup{
	colorlinks=true,
	linkcolor=mydarkblue,
	citecolor=mydarkblue,
	filecolor=mydarkblue,
	urlcolor=mydarkblue,
}
\usepackage{thmtools,thm-restate}
\usepackage[numbers]{natbib}
\usepackage{todonotes}
\usepackage{bm}
\usepackage{tikz-cd}
\usepackage{bbm}
\usepackage{placeins}

\date{\vspace{-5ex}}

\allowdisplaybreaks

\newtheorem{lemma}{Lemma}

\newtheorem{fact}{Fact}

\theoremstyle{definition}

\theoremstyle{remark}

\numberwithin{equation}{section}
\DeclareMathOperator*{\argmin}{arg\,min}
\DeclareMathOperator*{\argmax}{arg\,max}

\newcommand*{\poly}{\text{poly}}

\algnewcommand\algorithmicsend{\textbf{Break}}
\DeclareSymbolFont{bbold}{U}{bbold}{m}{n}
\DeclareSymbolFontAlphabet{\mathbbold}{bbold}

\begin{document}
\newcommand{\theTitle}{An Efficient Algorithm for High-Dimensional Log-Concave Maximum Likelihood}
\author{
 \fontsize{11}{13}\selectfont {\bf Brian Axelrod} \\
 \fontsize{11}{13}\selectfont Stanford University \\
 \fontsize{11}{13}\selectfont {\tt baxelrod@cs.stanford.edu}
\and
\fontsize{11}{13}\selectfont {\bf Gregory Valiant} \\
\fontsize{11}{13}\selectfont Stanford University \\
\fontsize{11}{13}\selectfont {\tt valiant@stanford.edu}
}

\title{\theTitle}
\date{}

\clearpage
\maketitle

\begin{abstract}
The log-concave maximum likelihood estimator (MLE) problem answers: for a set of points $X_1,...X_n \in \mathbb R^d$, which log-concave density maximizes their likelihood? We present a characterization of the log-concave MLE that leads to an algorithm with runtime $\poly(n,d, \frac 1 \epsilon,r)$ to compute a log-concave distribution whose log-likelihood is at most $\epsilon$ less than that of the MLE, and $r$ is parameter of the problem that is bounded by the $\ell_2$ norm of the vector of log-likelihoods the MLE evaluated at $X_1,...,X_n$.
\end{abstract}

\thispagestyle{empty}
\newpage
\setcounter{page}{1}	

\section{Introduction}
\subsection{Motivation and Related Work}
One of the central questions of both statistics and learning is recovering a distribution from samples. Much of the work in estimation and learning focuses on parametric statistics which assumes the data generating distribution is of a parametric form. This assumption often allows for statistically and computationally efficient inference and learning.

Shape-constrained density estimation aims to create a middle ground between parametric statistics and making no distributional assumptions. Instead of assuming the distribution is of a particular parametric form, shape-constrained density assumes the distribution has a density conforming to a shape constraint such as log concavity.

Log-concave density estimation in high dimensions, in particular, has been studied by both the learning and statistics communities. Cule et al. were the first to study the recovery of log-concave densities in high dimensions \cite{cule2010maximum}. They showed that the log-concave maximum likelihood estimator (MLE) converges asymptotically and proposed an algorithm to compute it \cite{cule2010maximum, cule2010theoretical}. Computational efficiency was not a focus of the work and the presented algorithm has a step which requires computing a large triangulation. For $n$ samples in $\mathbb R^d$, the triangulation can be of size $O(n^{d/2})$ making it difficult to scale the algorithm to large dimensions.

Later work characterized the finite sample complexity of log-concave density estimation. First Kim et al. showed that no method can get closer than squared Hellinger distance $\epsilon$ (and indirectly, total variation distance) with $\widetilde O\left ( 1/\epsilon^{\frac{d+1}{2}} \right ) $ samples, where $\widetilde O$ hides logarithmic factors in $1/\epsilon, d$ \cite{kim2016global}. Later work demonstrated methods for learning log-concave distributions in total variation and squared Hellinger distances with bounded sample complexity. First Diakonikolas et al. showed a method that obtains sample complexity $\widetilde O\left ( 1/\epsilon^{\frac{d+5}{2}} \right ) $ with respect to total variation distance \cite{diakonikolas2016learning}. This work did \emph{not} use the log-concave MLE. Carpenter et al. later showed that the log-concave MLE is also effective for learning in squared Hellinger distance \cite{carpenter2018near}. The log-concave MLE was shown to converge to square Hellinger distance $\epsilon$ with $\widetilde O\left ( 1/\epsilon^{\frac{d+3}{2}} \right ) $ with high probability, showing the log-concave MLE is nearly optimal in said metric. Both of these works are non constructive and do not provide efficient algorithms.

However, while the line of work studying the log-concave MLE from a information-theoretic perspective is extensive, there is little work on finding a \emph{computationally efficient} algorithm for high dimensional log-concave MLE problems. Our main contribution is a characterization of the solution to the log-concave MLE that leads to an algorithm with polynomial dependence on the dimension.

\begin{restatable}{theorem}{thmalgruntime}
Let $p^\star$ be the solution to the log-concave MLE and $l(p)$ be the log-likelihood of $X_1,...,X_n \in \mathbb R^d$. A distribution $p$ such that $l(p) - l(p^\star) \le \varepsilon$ can be computed with high probability in time $\poly(n,d,\frac 1 \varepsilon, r)$ where $r$ bounds the $\ell_2$ norm of the log-likelihoods of $X_1,...,X_n$ under $p^\star$.
\end{restatable}

It is important to note the magnitude of $n$ above. In the worst case, the value of $n$ must be at least $1 / \epsilon^{\frac{d+1}{2}}$ for the log-concave MLE to have converged information theoretically. The previous algorithm by Cule et al. must take at least $n^{\frac d 2} > 1/ \epsilon^{\frac{d^2}{4}}$ arithmetic operation. Our algorithm requires at most $O_\epsilon(1/\epsilon^{O(d)} poly(d))$ arithmetic operations where $O_\epsilon$ hides the radius. Note that since the previous algorithm was also a first order method, a standard analysis would have a similar dependence on $r$ (though no such analysis was provided).

\subsection{Overview}
The key insight underlying the efficient algorithm is a new geometric characterization of the solutions to the log-concave MLE. The solutions to log-concave MLE are contained within a class of distributions known as \emph{tent distributions}, whose log-likelihoods correspond to polyhedra \cite{cule2010maximum}. Our contribution lies in observing that while tent distributions are not an exponential family, they ``locally" retain many properties of exponential families. In fact, tent distributions can be viewed as the union of a finite collection of exponential families that share a log-partition function. This union preserves the following properties of the maximum likelihood geometry that makes maximum likelihood estimation tractable:
\begin{enumerate}
    \item The objective is convex.
    \item Samples from the distributions can be used to compute unbiased estimates of the gradient of the likelihood with respect to a natural parameterization.
\end{enumerate}
Finally, we show that the solution to the log-concave MLE is a particular solution to the tent-density maximum likelihood problem.

It is known that sample access to an exponential family leads to a simple stochastic gradient descent based algorithm for computing maximum likelihood estimates \cite{wainwright2008graphical}. The algorithm maintains a distribution (from the hypothesis class) at each iteration and generates a single sample from this distribution. The computational efficiency follows from the convexity of the the log-likelihood function and the fact that an unbiased estimate of the gradient can be computed from a sample of the distribution corresponding to the current iteration.

The ``exponential form" of tent distributions developed in this paper retains many properties of the exponential family. In section \ref{sec:tentgeo}, we show that the exponential form of tent distribution maximum likelihood also results in a convex optimization. In Section \ref{sec:suffstat}, we develop the notion of the polyhedral sufficient statistic. The polyhedral sufficient statistic allows us to compute the density of a point while only depending on the parameters through a combinatorial property of the tent function known as the regular subdivision. Over regions of the parameter space where the regular subdivision remains fixed, tent distributions form true exponential families. We show that since the same log-partition function is shared across all tent distributions, samples can be used to compute unbiased estimators of the gradient just as for exponential families. In practice this means that the same algorithm that we would expect to use for computing the MLE for exponential families can be used to compute the MLE for tent distributions.

However, the tent distribution MLE problem is \emph{not} the same as the log-concave MLE problem. In section \ref{sec:tentgeo}, we provide a characterization of the log-concave MLE as a particular solution to the tent distribution MLE problem. Beyond the computational implications, this characterization also helps motivate why the log-concave MLE makes sense as a method for log-concave density estimation. The log-concave MLE is the distribution with the uniform polyhedral sufficient statistic.

\subsection{Future Directions}
In this work we demonstrate:
\begin{enumerate}
    \item A faster algorithm for computing the log-concave MLE.
    \item That the machinery developed for exponential families can also be applied to a significantly broader class of distributions.
\end{enumerate}
The geometry of exponential family maximum likelihood estimation has been extensively studied.  For example, we have natural stochastic oracles for both the gradient and Hessian of the objective function (see e.g.~\cite{wainwright2008graphical}).  Despite this, we still lack algorithms that converge faster than at a sublinear rate.  It is certainly plausible that faster algorithms exist, perhaps either of the form of a stochastic Newton-type algorithm, or via a cutting plane method in an appropriate metric.

In this paper we show that many of the geometric properties that make exponential family maximum likelihood optimizations tractable also extend to a more general class of distributions.  This prompts two, interrelated questions: \emph{What is the broadest class of distributions that admits these properties?}  \emph{And which of the algorithmic tools developed for exponential families can be applied to this larger class?}

\subsection{Preliminaries}
We say that a probability density $p(x)$ is \emph{log-concave} if $\log p(x)$ is concave. The log-concave MLE of a set of points $X_1,...,X_n \in \mathbb R^d$ is the log-concave density $\hat p(x)$ such that $\prod\limits_i \hat p (X_i)$ is maximized.

We say that a probability density $f$ is $C-$isotropic if for any unit vector $u$:
$$\frac 1 C \le  \int (u^T x)^2 d \pi_f (x) \le C$$.

The indicator function for a set $X$ is denoted as follows: $\mathbbm 1_{X}(x) = \begin{cases}1 :& x \in X\\0:& x \not\in X\end{cases}$. For an natural number $n$, the all ones vector in $\mathbb R^n$ is denoted $\mathbbm 1_n $.  We say that a function $f$ is in $C^\infty$ if it is smooth (infinitely differential on its domain).  Throughout, We let $\langle x,y \rangle$ denote the inner product between real vectors $x,y$.

\section{The Geometry of Tent Distributions \label{sec:math}}
Our algorithm relies on a characterization of the geometry of log-concave distributions that maximize the likelihood of a point set. These solutions are always of a particular form, known as tent densities. We begin by introducing exponential families and defining tent densities. We then provide intuition for the algorithm by characterizing how tent distributions preserve the geometry that makes exponential family maximum likelihood estimation tractable.

\subsection{Exponential Families}
In this section we give a brief overview of exponential families that covers just the material necessary to understand this paper. If you are familiar with exponential families we advise you skip to section \ref{sec:tentdef}. If you are not, we recommend you read this section and reference \cite{wainwright2008graphical} for a more complete treatment of exponential families.

An \emph{exponential family} parameterized by $\theta \in \mathbb R^k$ with \emph{sufficient statistic} $T(x)$, with carrier density $h$ measurable and non-negative is a family of probability distributions of the form:
$$ p_\theta (x)  = exp (\langle T(x), \theta \rangle - A(\theta)) h(x)$$
The \emph{log-partition} function $A(\theta)$ is defined to normalize the integral of the density.
$$ A(\theta) = \log \int \exp( \langle T(x), \theta \rangle ) h(x) dx$$
It makes sense to restrict our attention to values of $\theta$ that give a valid probability density. The set of \emph{Canonical Parameters} $\Theta$ is defined such that $\Theta = \{ \theta \mid  A(\theta) < \infty\}$.

We say that an exponential family is \emph{minimal} if $\theta_1 \neq \theta_2$ implies $p_{\theta_1} \neq p_{\theta_2}$. This is necessary and sufficient for statistical identifiability.

We will study the geometry of maximum likelihood estimation for exponential families.

The maximum likelihood parameters $\theta^\star$ for a set of iid samples $X_1,...X_n$ are:
\begin{align}
    \theta^\star &= \argmax_\theta \prod\limits_i p_\theta(X_i) \nonumber \\
                 &= \argmax_\theta \log \prod\limits_i p_\theta(X_i) \nonumber \\
                 &= \argmax_\theta \sum\limits_i \langle T(X_i), \theta \rangle - nA(\theta) - \sum\limits_i \log h(x_i)\nonumber \\
                 &= \argmax_\theta   \left \langle \frac 1 n \sum\limits_i T(X_i), \theta  \right \rangle - A(\theta) \label{opt:explike}
\end{align}
We refer to the optimization in equation \eqref{opt:explike} as the \emph{exponential maximum likelihood optimization}. The last equation helps highlight why $T(x)$ is referred to as the sufficient statistic. No other information is needed about the data points to compute both the likelihood and the maximum likelihood estimator.

One reason why exponential families are important is that the geometry of the optimization in equation \eqref{opt:explike} has several nice properties.
\begin{fact}
$A(\theta)$ satisfies the following properties:
\begin{enumerate}
    \item $A(\theta) \in C^\infty$ on $\Theta$.
    \item $A(\theta)$ is convex.
    \item If the exponential family is minimal, $A(\theta)$ is strictly convex.
    \item $\nabla A(\theta) = \mathbb E_{x \sim p(\theta)} [T(x)]$.
\end{enumerate}
\end{fact}

While the distributions we use in this paper are \emph{not} an exponential family, we will show that the corresponding optimization retains all the properties described above except the smoothness. These properties will be the fundamental building blocks of the efficient algorithm presented in section \ref{sec:alg}.

\subsection{Tent Distributions\label{sec:tentdef}}
In this section we define the notation necessary to work with tent densities. Tent densities are notable because the solution of a log-concave maximum likelihood estimation problem is always tent density \cite{cule2010maximum}.

We define tent functions (and, later, subdivisions) using the notation due to \cite{robeva2017geometry}. Take any $X_1,...X_n \in \mathbb R^d$ and corresponding $y_1,...y_n \in \mathbb R$. We refer to the matrix with columns $X_i$ as $X$ and the vector with elements $y_i$ as $y$. The \emph{tent function} $h_{X,y} : \mathbb R^d \rightarrow \mathbb R$ is the pointwise smallest concave function such that $h_{X,y}(X_i) = y_i $. The points $(X_i, y_i)$ are referred to as \emph{tent poles}. Note that the function $h$ is $-\infty$ outside of the convex hull of $X_1,...X_n$ and its graph is a polytope. See figure \ref{fig:tent2d} for a side by side tent density and tent function. See figure \ref{fig:tent} for the graph of an example tent function.

When $p_{X,y}(x) = \exp(h_{X,y}(x))$ integrates to one, we refer to it as a \emph{tent density} and the corresponding distribution as a \emph{tent distribution}. The support of a \emph{tent distribution} must be within the convex hull of $X_1,...X_n$.

\begin{figure}
    \centering
    \includegraphics[width=0.5\linewidth]{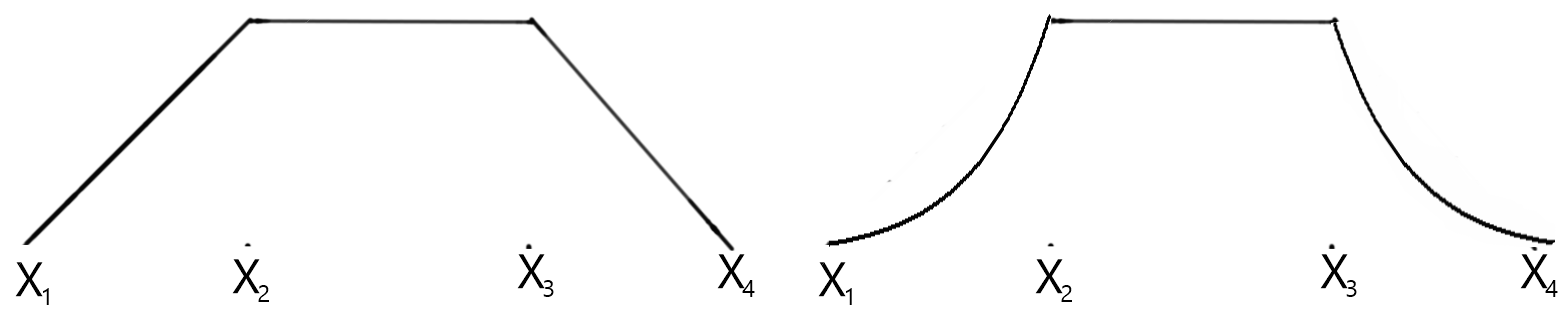}
    \caption{A $2D$ tent function and the corresponding tent density side by side. The two functions are not plotted to scale. }
    \label{fig:tent2d}
\end{figure}

\begin{figure}
    \centering
    \includegraphics[width=\linewidth]{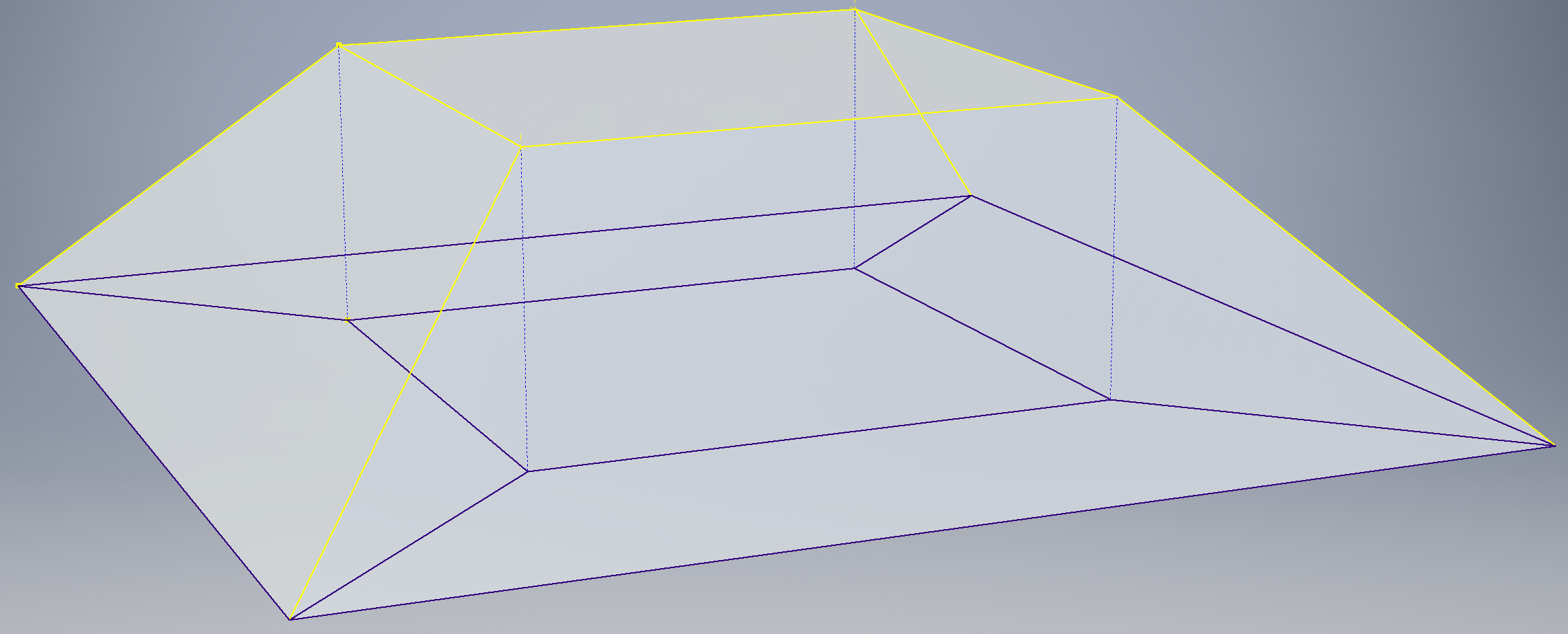}
    \caption{An example of a tent function and its corresponding regular subdivision. }
    \label{fig:tent}
\end{figure}

Recall that tent densities are notable because contain solutions to the log-concave MLE. Consider the log-concave maximum likelihood estimation problem over $X_1,...X_n$. The solution is always a tent-density because tent densities with tent poles $X_1,...,X_n$ are the minimal log-concave functions with log densities $y_1,...y_n$ at points $X_1,...X_n$. A function that was not a tent function would waste density on points that would not improve the likelihood score used in the optimization.

The algorithm which we present can be thought of as an optimization over tent functions. In section \ref{sec:suffstat} we will show a parametric form of tent distributions that looks very similar to an exponential family and suggests that tent distributions retain many important properties of exponential families.

\subsection{Exponential Families and the Polyhedral Sufficient Statistic \label{sec:suffstat}}
In this section we will compare tent distributions to exponential families by defining a sufficient statistic that makes tent distributions ``locally" exponential families. In order to define the sufficient statistic we need to understand the regular subdivision induced by a tent function.

Given a tent function $h_{X,y}$ with $h_{X,y} (X_i) = y_i$, its associated \emph{regular subdivision} $\nabla_y$ of $X$ is a collection of subsets of $X_1,...X_n \in \mathbb R^n$ whose convex hulls are the regions of linearity of $h_{X,y}$. See Figure \ref{fig:tent} for an illustration of a tent function and its regular subdivision. We refer to these polytopes of linearity as \emph{cells}. We say that $\delta_y$ is a \emph{regular triangulation} of $X$ if every cell is a $d-$dimensional simplex.

It is helpful to think of regular subdivisions in the following way: Consider the hyperplane $H$ in $\mathbb R^{d+1}$ obtained by fixing the last coordinate. Consider the function $h_{X,y}$ as a polytope and project each face onto $H$. Each cell is a projection of a face, and together the cells partition the convex hull of $X_1,...,X_n$.  Observe that regular subdivisions may vary with $y$. Figure \ref{fig:changediv} provides one example of how changing the $y$ vector changes the regular subdivision.

\begin{figure}
    \centering
    \includegraphics[width=\linewidth]{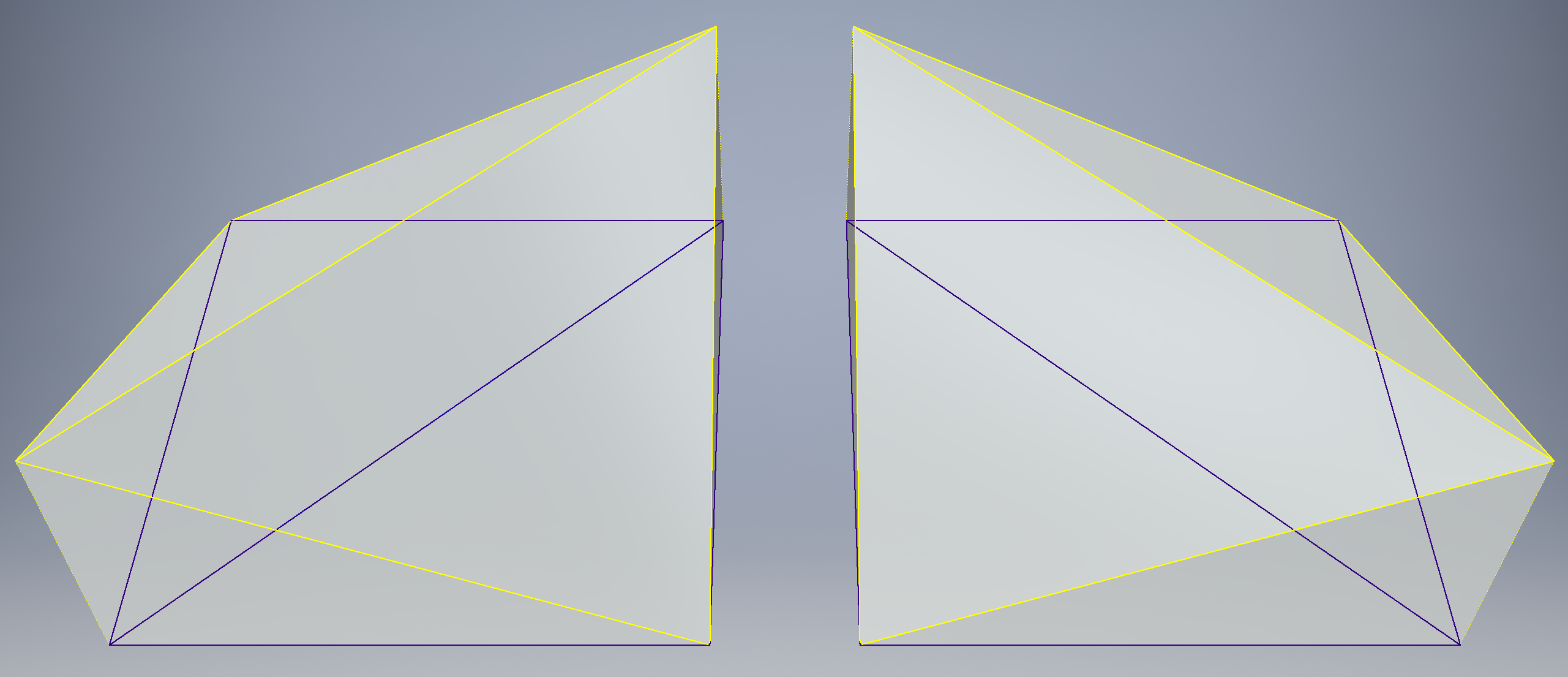}
    \caption{Changing the height of the tent poles can change the induced regular subdivision (shown in purple). }
    \label{fig:changediv}
\end{figure}

For a given regular triangulation $\nabla$, the associated \emph{consistent neighborhood} $N_{\nabla}$  is the set of all $y \in N$,  such that $ \nabla_{y} =\nabla$.  That is, consistent neighborhoods are the sets of parameters where the \emph{regular triangulation} remains fixed. Note that these neighborhoods are open and  their closures cover the whole space.
We note that when $y$ is chosen in general position $\nabla_y$ is always a regular triangulation.

Consider a regular triangulation $\nabla$. The \emph{polyhedral statistic} is the function $$T_y(x): CONV\_HULL(X_1,...X_n) \rightarrow [0,1]^n,$$ that expresses $x$ as a convex combination of corners of the cell containing $x$ in $\nabla_y$. That is $x = X T_y(x)$ where  $||T_y(x)||_1 = 1$ and $T_y(x)_i = 0$ if $X_i$ is not a corner of the cell containing $x$. The polyhedral statistic gives an alternative way of writing tent functions and tent densities:
$$ h_{X,y}(x) = \langle T_y(x), y \rangle$$
$$ p_{X,y}(x) = \exp( \langle T_y(x), y \rangle)$$

If we restrict $\theta$ such that $\sum\limits_i \theta_i = 1$ and define $A_y(\theta) = \log \int\limits_x p_{X,y}(x) dx$ then we can see that for every consistent neighborhood $N_\nabla$ we have an exponential family of the form
\begin{align}
     \exp\left ( \langle T_\theta(x), \theta \rangle - A(\theta) \right ) \; \textrm{ for }  \theta \in N_\nabla. \label{eqn:expfam}
\end{align}
While equation \eqref{eqn:expfam} shows how subsets of tent distributions are exponential families, it also helps highlight why tent distributions are \emph{not} an exponential family. The sufficient statistic depends on $y$ through the regular subdivision. This means that tent distributions do not admit the same factorized form as exponential families since the sufficient statistic depends on $y$.

Note that we can use any ordering of $X_1,\ldots,X_n$ to define the polyhedral sufficient statistic everywhere including on regular subdivisions that are \emph{not} regular triangulations. Also note that eliminating the last coordinate using the constraint $\mathbbm 1_n^T \theta = 1$ makes each exponential family minimal. In other words, over regions where the regular subdivision does not change (for example the consistent neighborhoods), tent distributions are minimal exponential families. This means the set of tent distribution can be seen as the finite union of a set of minimal exponential families. We refer to equation \eqref{eqn:expform} as the exponential form for tent densities.
\begin{equation}
    p_{X,y} (x) = \exp\left ( \langle T_y(x), y \rangle - A(y) \right ) \mathbbm 1_{CONV\_HULL(X_1,...,X_n)}(x). \label{eqn:expform}
\end{equation}

\subsection{The Geometry of the Tent Distribution MLE \label{sec:tentgeo}}
Understanding tent distributions as ``almost" being an exponential family is exactly what enables the efficient algorithm for computing the log-concave MLE.

Consider the following optimization over tent distributions for some fixed vector $\mu \in [0,1]^n, ||\mu||_1 = 1$. Note the similarity to the exponential family maximum likelihood optimization.
\begin{equation}
    \theta^\star = \argmax\limits_{\theta} \left (\langle \mu, \theta \rangle - A(\theta) \right )
\end{equation}

By definition, $A (\theta) = \log \int \exp (\langle \mu, \theta \rangle) dx$, which is, up to a linear component, exactly the logarithm of the function $\sigma$ studied in \cite{cule2010maximum}. We use the fact that $\exp(A(\theta))$ is convex combined with the exponential family geometry to prove that $A(\theta)$ is convex.   This allows us to prove that the geometry of the log-partition function of tent distributions behaves, in many ways, similarly to that of exponential families.
\begin{restatable}{theorem}{thmtentexp}\label{thm:convex}
 $A(\theta)$ is convex and $\mathbb E_{x \sim p_\theta}[T(x)] \in \partial_\theta A(\theta)$.
\end{restatable}
Please see Appendix \ref{app:proofs} for a proof.

We are now able to characterize the solution to the log-concave MLE in terms of the polyhedral sufficient statistic.  Let $\hat \theta$ be the parameters of the log-concave MLE expressed in the exponential form of tent distributions. Then
\begin{align}
   \hat \theta &= \argmax\limits_\theta \left \langle \frac 1 n \sum\limits_i T(X_i), \theta \right \rangle   - A(\theta) \nonumber\\
   &= \argmax\limits_\theta \left \langle \frac 1 n \mathbbm{1}, \theta \right \rangle  - A(\theta) \label{opt:mletentmath}
\end{align}
That is, the log-concave MLE is the tent function with polyhedral sufficient statistic equal to $\frac 1 n \mathbbm{1}$. This means we can think of the log-concave MLE as the the tent distribution which has most even support over its tent poles.  Furthermore, we  observe that since $\mathbbm 1^T \theta = 1$, equation \eqref{opt:mletentmath} is equivalent to simply minimizing the log partition function.

\subsection{Tent Distributions and Geometric Combinatorics}
The geometric aspect of this paper can be thought of as a continuation of the work started by in \cite{robeva2017geometry}. They observe that solving the log-concave  MLE problem over a \emph{weighted} set of samples has a strong connection to geometric combinatorics.

They examined the relationship between the weight vector and the combinatorial structure of the tent density maximizing the weighted likelihoods. In particular, they showed that different values of the weight vector can induce every regular subdivision.
In the context of our work, the weight vector can be thought of as an expected value of the polyhedral statistic. Then their work can be interpreted as studying the map from the expected value of the polyhedral sufficient statistic to regular subdivisions. The consistent neighborhoods are exactly the dual of the preimages of regular triangulations under this map.

Also of interest is the Samworth body studied by Robeva et al. The Samworth Body $S(x)$ is defined as follows:
$$ S(X) = \{ y \mid  \int \exp (h_{X,y} (x)) dx \le 1 \}$$
Points on the boundary of the Samworth body correspond to tent densities. This means that parameters of the exponential form of tent distributions are in bijection with the boundary of the Samworth body using the following map $y \rightarrow y - A(y)$.

\subsection{Algorithmic Toolkit \label{sec:toolkit}}
The algorithm presented in Section \ref{sec:alg} requires sample access to a tent distribution. Implementing this requires several oracles for tent distributions which are briefly describe below. Every oracle is implemented in detail in Appendix \ref{app:oracles}.

\subsubsection{Relative Density Queries}
The unscaled density $A(y) p_{X,y}(x) $ can be computed using a linear program. Recall that a tent function is defined as the minimal concave function containing the tent poles. This means the density value at a particular point may be computed by finding the largest value still in the convex hull of the tent poles.

\subsubsection{Polyhedral Statistic}
The polyhedral statistic can also be computed using the same linear program used to compute density queries. Instead of using the $y$ value of the point, we use the vector used to express it as a convex combination of tent poles. This is sufficient for $y$ chosen in general position resulting in $\nabla_y$ being a regular triangulation.

Even when $y$ is not chosen in general position, it is possible to modify the procedure to return a particular polyhedral statistic. However, this is not strictly necessary for the correctness of the algorithm we present since any valid polyhedral statistic will be a subgradient. In fact taking the union over different ways to compute the statistic will give us a set that spans the subgradients.

\subsubsection{Line-Restricted Sampling}
In order to be able to generate samples from a tent distribution we will have to sample along the tent density restricted to an arbitrary line. It turns out that restricting the tent density to a line yields another tent distribution with at most $n$ tent poles. We combine this with a characterization of the measure of pieces of a $2D$ tent distribution from \cite{cule2008auxiliary} to create an exact sampler for this restricted distribution.

\subsubsection{Sampling from Tent Functions}
In order to be able to generate samples from tent functions we use a hit and run random walk with analysis from \cite{lovasz2007geometry}.
\begin{enumerate}
    \item Choose a line uniformly at random from the current point.
    \item Sample from the density restricted to this line
    \item return to step 1
\end{enumerate}

Making this random walk mix quickly requires finding a transformation that makes the tent density almost isotropic (equivalent to computing its second moment). This and a complete analysis of the random walk are presented in Appendix \ref{app:oracles}.

\subsubsection{Evaluating the Log-Partition Function}
It is important that we use a different representation of tent distributions than used by \cite{cule2010maximum}. Our representation has a strong connection to exponential families and readily admits many important queries including the following:
\begin{enumerate}
    \item Unscaled Density Queries
    \item First and Second Moments
    \item Sampling
    \item Linearly-Restricted Distributions
\end{enumerate}

However, for completeness, we include a method to evaluate an additive approximation  of the log-partition function. Converting from the exponential form to the representation used by Cule et al. requires adding the log-partition function to the parameter vector.

To evaluate the log-partition function we imitate Lebesgue integration of the tent function. We take advantage of the fact that the density is quasi-concave. We divide up the graph into thin slices and compute the volume of each slice. Adding them up yields an approximation of the volume. We use a property of log-concave distributions to bound the number of slices necessary to compute a good approximation.

\section{Method \label{sec:alg}}
Our algorithm works by applying stochastic gradient descent to the optimization in equation \eqref{opt:explike} with a uniform polyhedral statistic. The stochastic gradient is computed by evaluating the polyhedral statistic on a sample from the current tent function. The psuedocode is presented in \ref{sec:pseudocode}. Recall that these samples can be generated within total variation distance $\epsilon$ in time $\poly(n,d, \frac 1 \epsilon)$ using the oracle described in Appendix \ref{app:oracles}.  Section \ref{sec:analysis} will bound the convergence rate of this stochastic gradient descent.

\subsection{Pseudocode\label{sec:pseudocode}}
The algorithm itself is quiet simple and described in algorithm \ref{alg:complete}.
\begin{algorithm}
  \caption{Compute the log-concave maximum likelihood \label{alg:complete}}
  \begin{algorithmic}[1]
	\Function{ComputeLogConcaveMLE}{$X_1,...X_n, m$}
	\State $y \gets \frac 1 n\mathbbm 1_n$
	\For{$i\gets 1,m$}
	    \State $\eta \gets  1 / \sqrt i$
	    \State $s \sim p_{X,y}$
	    \State $y \gets y + \eta \left  (\frac 1 n\mathbbm 1_n - T_y(s) \right)$
	\EndFor
	\State \Return $y$
	\EndFunction
    \Statex
  \end{algorithmic}
\end{algorithm}
Note that two steps are abstracted away above: sampling from the tent function and computing the polyhedral sufficient statistic. Sampling from the tent function can be done using a hit and run random walk. This is described briefly in section \ref{sec:toolkit} and in detail in Appendix \ref{app:oracles}. The polyhedral sufficient statistic can be computed using a linear program. This linear program is summarized in section \ref{sec:toolkit} and described in detail in Appendix \ref{app:oracles}.

\subsection{Main Analysis \label{sec:analysis}}
In this section we prove the main theorem.

\thmalgruntime*

Recall that algorithm \ref{alg:complete} is simply applying stochastic gradient descent to the following function:
$$h(y) = \left \langle \frac 1 n \mathbbm 1_n, y \right \rangle - A(y) $$
Recall from theorem \ref{thm:convex} that $h$ is convex. In order to analyze the convergence rate we rely on an analysis due to \cite{shamir2013stochastic}.

\begin{lemma}[\cite{shamir2013stochastic}] \label{lemma:sgd}
Suppose that $F$ is convex with domain $\mathcal W$, $g_t$ is the stochastic gradient at iteration $t$, and that for some constants $D, G$, it holds that $\mathbb E[||g_t||] \le G^2$ for all $t$, and $\sup_{w, w' \in \mathcal W} || w - w'|| \le D$ Consider stochastic gradient descent with step sizes $\eta_t = c/\sqrt t$ where $c > 0 $ is a constant. Then for any $T> 1$, it holds that
$$ \mathbb E [F(w_T) - F(w^\star)] \le \left ( \frac{D^2}{c} + c G^2 \right )\frac{2 + \log T}{\sqrt T} $$
\end{lemma}

This allows us to prove the main theorem.
\begin{proof}
Lemma \ref{lemma:sgd} captures almost everything we need. All that is left is to  $G$ and specify the necessary total variation distances for the hit-and-run sampling.

Notice that at every point the stochastic gradient $g_t = \frac 1 n \mathbbm 1_n - T(x)$ for some $x$. The norm of said quantity, $|\frac 1 n \mathbbm 1_n - T(x)|$ is maximized when $T(x) = e_i$ for some $i$.
Thus $\mathbb E \left [\frac 1 n \mathbbm 1_n - T(x) \right ] < \sup\limits_x \left |\frac 1 n \mathbbm 1_n - T(x) \right | = O(1)$.

The above lets us compute the number of iterations necessary. If $m$ iterations are necessary, running each hit and run to total variation distance $\frac 1 {m^2}$ would mean that with probability $\frac 1 m$ the stochastic gradients would be indistinguishable from the true distribution and the algorithm would succeed with high probability. If the number of iterations is not known apriori, one could guess the number of required iterations. If the guess is too low, the algorithm can be restarted with a factor 2 higher iteration limit. This can be repeated until the desired convergence is achieved with only a constant factor slowdown comparing to knowing the necessary number of iterations in advance.
\end{proof}


\section*{Acknowledgements}
This work was supported by NSF Fellowship grant DGE-1656518, NSF awards CCF-1704417, CCF-1763299 and an ONR Young Investigator Award (N00014-18-1-2295).

\newpage

\appendix
\section{Oracle Implementations \label{app:oracles}}
In this appendix we give a complete description of every oracle described in Section \ref{sec:toolkit}.

\subsection{Density Queries}
The density $p_{X,y}(x)$ can be computed using the following packing linear program:
\begin{align}
&\max y \label{opt:densitylp}\\
x &=\sum\limits_i \alpha_i X_i \nonumber\\
y&=\sum\limits_i \alpha_i y_i  \nonumber\\
1&= \sum\limits_i \alpha_i \nonumber\\
\alpha_i &\ge 0 \nonumber
\end{align}
The the point $y^\star$ that achieves the optimum of \eqref{opt:densitylp} can be used to compute the density as $p_{X,y}(x) = \exp(y^\star)$.

\subsection{Polyhedral Statistic}
The Polyhedral Statistic can also be computed using linear program \eqref{opt:densitylp} by simply returning $\alpha$. Note that if the $y$'s are in general position there will be a unique solution to the linear program.

If the linear program does not have a unique solution and the solver return $\alpha$ with more than $d$ nonzero entries, one can simply add a constraint that forces one of the nonzero entry to be zero and repeat the process. At every stage before a simplex is obtained there exists at one entry for which the linear program will remain solvable.

\subsection{Line-Restricted Sampling}
In this section we give a method for sampling from the distribution with density proportional to $g(x) = \exp (h_{X,y}(x_0 + t \theta)) $ for $x_0$ and $\theta$ chosen in general position.

First note that $g(x)$ is itself a 1-dimensional unscaled tent-density. Since $\theta$ induces an order on tent poles of $h$ (by sorting via $\theta^T X_i$), $g$ has at most $n$ tent poles. We can compute these tent poles and then sample exactly by computing the measure of each segment between tent poles (see formula in \cite{cule2008auxiliary}), sampling a segment, and then sampling on the distribution restricted to the segment. The psuedocode for this process can be in Algorithm \ref{alg:linsampler}.

\begin{algorithm}[H]
  \caption{Sample from $\frac{\exp (h_{X,y}(x_0 + t \theta)) }{\int \exp (h_{X,y}(x_0 + t \theta))dt }$
    \label{alg:linsampler}}
  \begin{algorithmic}[1]
	\Function{sample}{$X_1,...X_n, y_1,...y_n$}
	\State $t' \gets \argmin\limits_t x_0 + t\theta \in CONV\_HULL(X_1,...X_n)$
	\State $z_0 \gets x_0 + t'\theta $ \Comment{$z_0,...z_m, m \le n$ are the tent poles of $g$}.
	\For{$i \gets 1,n$}
	    \State $\beta_j \gets \begin{cases}1&: T(z_{i-1} + \epsilon \theta)_j > 0 \\ 0&: \textrm{otherwise} \end{cases}$ for $j$ $1,...,n$.
	    \State $t' \gets \argmax\limits_{t } z_{i-1} + t \theta = \alpha^T X \beta $ s.t. $\alpha\in[0,1]^m, ||\alpha||=1$.
	    \State $z_i \gets z_{i-1} + t \theta$
	    \If {$z_i \not\in int(CONV\_HULL(X_1,...X_n)$}
	    	\State $m \gets i-1$
	        \State \textbf{break}
	    \EndIf
	\EndFor
	\State $\alpha_i \gets (||z_i - z_{i+1}||)\frac{|g(z_i) - g(z_{i+1})|}{|\log g(z_i) - \log g(z_{i+1})|}$ for $i \gets 0,{m-1}$
	\State $j \gets i$ with probability $\frac{\alpha_i}{\sum\limits_i \alpha_i}$
	\State $p_1 \sim exp(\log g(x_j))$
	\State $p_2 \sim exp(\log g(x_{j+1}))$
	\State \Return $\frac{p_1}{p_1 + p_2}x_j +\frac{p_2}{p_1 + p_2}x_{j+1} $
	\EndFunction
    \Statex
  \end{algorithmic}
\end{algorithm}
\FloatBarrier

\subsection{Sampling from Tent Functions}
In order to be able to generate samples from tent functions we use a hit and run random walk with analysis in \cite{lovasz2007geometry}.
\begin{enumerate}
    \item Choose a line uniformly at random from the current point.
    \item Sample from the density restricted to this line
    \item return to step 1
\end{enumerate}

Consider a tent density $f$ over $\mathbb R^d$ that is at most $C-$isotropic. Let the initial point, $x_0$ be drawn from a distribution at most total variation distance $H$ from $\pi_f$.
Then the distribution of the $m$th sample of the above random walk will be at most total variation distance $\epsilon$ from $\pi_f$ if
$$ m \ge O \left (C^4 H^4 \frac{n^3}{\epsilon^4} \ln^3 \frac{2H}{\epsilon} \right )$$

However, we have no reason to believe that our tent density will be close to isotropic and we pay a polynomial factor in $C$ above. We alleviate this issue by ``rounding" the density \cite{lovasz2007geometry}. That is we find a linear operator $W$ that transforms the density into approximately isotropic position.

Rounding will require a level set separation oracle. The following idea will be sufficient to design the separation oracle for the level set: identify a face of the tent function that separates the point in question and compute it's intersection with the plane of the level set.

Let $\overline X_i = <X_i^1,...X_i^d, y_i>$, and let $y_{max} = \max\limits_i y_i$ be the highest log-density and $X_{max}$ the corresponding tent pole. To compute a separating hyperplane between the a point $Z$ and the $\exp(y_{max}) \delta$ level set compute the following linear program:
\begin{align*}
    \max t\\
    \kappa &= X_{max} + t (Z - X_{max})\\
    \kappa &= \sum\limits_i \alpha_i X_i \\
    1 &= \sum\limits_i \alpha_i\\
    \alpha_i &\ge 0 \; \forall i\\
    \sum\limits_i \alpha_i y_i &\ge y_{max} + \log \delta
\end{align*}
The $\alpha$ vector can be used to identify a $d+1$ dimensional face of the tent function. The intersection of that hyperplane and the hyperplane defined by setting the last coordinate to $y_{max} + \log \delta$ provides a separating hyperplane.  Given this oracle, we can now apply the rounding algorithm \cite{lovasz2007geometry}.

\subsection{Evaluating the Log-Partition Function}
Recall the definition of $A(y)$:
$$ A(y) = \log \int \exp( \langle T(x), y \rangle )  dx$$
Let $p_{max}$ be the maximum value of the unscaled tent density and $h_{max}$, its logarithm, the max of the tent function. In this section, we approximate a Lebesgue integral of $\int \exp( \langle T(x), y \rangle ) $ using slices $[p_{max} (1 - \epsilon)^i, p_{max} (1 - \epsilon)^{i+1})$. We will aim for a $(1 - \epsilon)$ approximation of this volume.

First, however, we must truncate the distribution. For any log-concave \emph{density} $f$ with maximum $M_f$:
$$\int\limits_{f(x) \le M_f \exp (-z)} f(x) dx \le \frac{\epsilon}{2}$$
for any $z \ge 2 \log (2/\epsilon) + d \log (O(d))$ due to Lemma 3.2 by \cite{carpenter2018near}.
This allows us to not worry about tiny level sets since the set of points where the density is low does not contribute much to the integral.

Now, for $i \in \mathbb N$ we examine the corresponding slice and level set, defined as $L_i = \{ x \mid h_{X,y}(x) \ge h_{max} - i\log ( 1 - \epsilon/2) \}$. This is the same as examining the $p_{max} (1 - \epsilon/2)^i$th level set of the unscaled density. By the above lemma is suffices to examine $i \in \left [0,...\left \lceil \frac{-z}{\log (1 - \epsilon/2)} \right \rceil\right ]$. Note that since $\epsilon$ is small (say less than 0.1), we can use the Taylor approximation of the logarithm to show that $O( \poly(1/ \epsilon, d) )$ intervals suffice. Note that the volume of the level set can be computed with high probability in polynomial time using the separation oracle presented earlier and a standard volume algorithm \cite{kannan1997random}. Then the lower approximation of the lesbesgue integral gives us an approximation with a corresponding bound on its quality:
\begin{align*}
   ( 1- \epsilon) \int f(x) dx  &\le ( 1- \epsilon/2) \int\limits_{f(x) \ge p_{max} \exp (-z)} f(x) dx\\
   &\le \sum\limits_i Vol(L_i) p_{max} ((1 - \epsilon/2)^i - (1 - \epsilon/2)^{i+1})\\
   & \le \int f(x) dx
\end{align*}
We note that the $(1 - \epsilon)$ multiplicative approximation of the integral gives us an $\epsilon$ additive approximation of the log partition function.

\newpage

\section{Proofs \label{app:proofs}}
\thmtentexp*

\begin{proof}
Recall that since $A(\theta)$ agrees with log-partition functions on consistent neighborhoods, so $A(\theta)$ must be smooth and convex on the consistent neighborhoods. We also use the continuity of $A$.

Now assume for the sake of contradiction that $A (\theta)$ is non-convex and there exists some $\theta_1, \theta_2$ s.t. $A(\theta_1 + t \theta_2)$ is not convex. Let $h$ denote this one dimensional function and let $g = \exp(h)$ be the same slice of $\exp(A(\theta))$. Since the closure of the consistent neighborhoods covers the parameter space, if $h$ is non-convex there must be a point $x$ at which $h$ is non-convex. Let $s_1 = \lim\limits_{y \to x^-} f'(y)$ and $s_1 = \lim\limits_{y \to x^+} f'(y)$. These limits must exist because $A$ is smooth on consistent neighborhoods and their closures cover the space. Since $h$ is nonconvex $s_1 > s_2$.

Now consider the equivalent limits $s_1', s_2'$ of $g$. Note that on the smooth parts of the domain $g'(x) = \exp (f(x))f'(x)$ so $s_1' = \exp(f(x)) s_1$ and $s_2' = \exp(f(x)) s_2$. Convexity of $g$ implies that $s_1' \le s_2' $ which contradicts $s_1 > s_2$.

Now we prove the fact about the subgradients of the log-partition function.  Note that since $A$ agrees with a log-partition function on the consistent neighborhoods, this holds immediately for $\theta$ in general position. The complement of consistent neighborhoods is also exactly the region in which there are multiple ways to define $T$. In fact, the set of valid $T$s spans the set of subgradients.

Compute $T$ using any appropriate, but fixed, triangulation (i.e. such that $T$ is $d-$sparse) and let $y$ be a direction towards the consistent neighborhood corresponding to this triangulation. Then $\lim\limits_{\epsilon \to 0+}\nabla_\theta A(\theta + \epsilon y) = \mathbb E_{x \sim p_\theta}[T(x)]$ and $T(x)$ is a subgradient in expectation.

In other words, since the triangulation corresponds to a consistent neighborhood adjacent to the current point, the chosen subgradient is the limit of the gradient when approaching the current point from that consistent neighborhood.
\end{proof}
\newpage

\bibliographystyle{alpha}
\bibliography{main}

\end{document}